\documentclass[12pt,a4paper]{elsarticle}

\usepackage{amsthm,amsmath,amssymb,url,longtable,booktabs,multicol,multirow,lineno}
\usepackage[classfont=sanserif,langfont=caps,funcfont=italic]{complexity}
\usepackage[multiple]{footmisc}
\usepackage{graphicx}
\usepackage{tikz}
\usepackage{subfig}

\allowdisplaybreaks

\newclass{\Wone}{W[1]}
\newclass{\Wstarone}{W^{\ast}[1]}
\newclass{\Wtwo}{W[2]}
\newclass{\Wstartwo}{W^{\ast}[2]}
\newclass{\Wthree}{W[3]}
\newclass{\Wstarfive}{W^{\ast}[5]}
\newclass{\Wfive}{W[5]}
\newclass{\Weight}{W[8]}
\newclass{\Wt}{W[t]}
\newclass{\Wstart}{W^{\ast}[t]}
\newclass{\Wtwotminustwo}{W[2t-2]}
\newclass{\WP}{W[P]}
\newclass{\paraNP}{para\text{-}NP}

\newlang{\MCSigmaTwoOne}{MC($\Sigma_{2,1}$)}
\newlang{\MCSigmatu}{MC($\Sigma_{t,u}$)}
\newlang{\MCSigmaStarFiveOne}{MC($\Sigma^{\ast}_{5,1}$)}
\newlang{\MCSigmaStartu}{MC($\Sigma^{\ast}_{t,u}$)}
\newlang{\PI}{Pattern Identification}
\newlang{\PIS}{PI with Small Strings}
\newlang{\PIP}{PI with Large Patterns}
\newlang{\PIPS}{PI with Large Patterns and Small Strings}
\newlang{\PIp}{PI with Small Patterns}
\newlang{\PIps}{PI with Small Patterns and Small Strings}
\newlang{\DS}{Dominating Set}
\newlang{\PVC}{Planar Vertex Cover}
\newlang{\SetCover}{Set Cover}
\newlang{\VC}{Vertex Cover}
\newlang{\WSATdCNFPos}{WSAT(\ensuremath{d}-CNF\ensuremath{^{+}})}
\newlang{\kFS}{\ensuremath{k}-Feature Set}

\newtheorem{theorem}{Theorem}[section]
\newtheorem{lemma}[theorem]{Lemma}
\newtheorem{corollary}[theorem]{Corollary}
\newtheorem{conjecture}[theorem]{Conjecture}
\theoremstyle{remark}
\newtheorem{claim}[theorem]{Claim}
\theoremstyle{definition}
\newtheorem{definition}[theorem]{Definition}

\newcommand{\problem}[3]{\begin{quote}\noindent\textsc{#1}:\\
                                    \begin{tabular}{rp{0.6\textwidth}}
                                      \textit{Instance: } &#2\\
                                      \textit{Question: } &#3
                                    \end{tabular}\end{quote}}

\newcommand{\Card}[1]{\left\vert #1 \right\vert}

\newcommand{\Yes}{\textsc{Yes}}
\newcommand{\No}{\textsc{No}}

\newcommand{\ie}{\emph{i.e.}}

\begin{document}

\begin{frontmatter}

\title{Separating Sets of Strings by Finding Matching Patterns is Almost Always Hard}
\author{Giuseppe Lancia} 
\address{Dipartimento di Matematica e Informatica, University of Udine, Via delle Scienze 206, 33100 Udine, Italy}
\author{Luke Mathieson}
\address{School of Electrical Engineering and Computer Science, University of Newcastle, Callaghan, NSW 2308, Australia}
\author{Pablo Moscato}
\address{School of Electrical Engineering and Computer Science, University of Newcastle, Callaghan, NSW 2308, Australia}

\begin{abstract}
We study the complexity of the problem of searching for a set of patterns that separate two given sets of strings. This problem has applications in a wide variety of areas, most notably in data mining, computational biology, and in understanding the complexity of genetic algorithms. We show that the basic problem of finding a small set of patterns that match one set of strings but do not match any string in a second set is difficult (\NP{}-complete, \Wtwo{}-hard when parameterized by the size of the pattern set, and \APX{}-hard). We then perform a detailed parameterized analysis of the problem, separating tractable and intractable variants. In particular we show that parameterizing by the size of pattern set and the number of strings, and the size of the alphabet and the number of strings give \FPT{} results, amongst others.
\end{abstract}


\end{frontmatter}

\section{Introduction}
\label{sec:intro}

Finding patterns in a collection of data is one of the fundamental problems in data mining, data science, artificial intelligence, bioinformatics and many other areas of both theoretical and applied computer science. Accordingly there are a large number of formulations of this problem. In this paper we develop a particular formulation, drawn from two central motivations:

\begin{enumerate}
\item multiparent recombination in genetic and evolutionary algorithms, and
\item the construction of explanatory patterns in single-nucleotide polymorphisms related to disease.
\end{enumerate}

It should not be construed however that these motivations are limitations on the applicability of the problem we develop and study. As will be seen, the underlying computational problem is a general one that occurs as a fundamental component of many other computational problems.

\subsection{The Central Problem}
\label{sec:central_problem}

Before expanding upon the motivations, we briefly introduce the core computational problem to provide a semi-formal context and some unifying vocabulary. For full definitions we refer the reader to Section~\ref{sec:prelims}. Central to the problem is the notion of \emph{pattern}, a string over an alphabet $\Sigma$ which has been augmented with a special symbol $\ast$. A pattern matches a string over $\Sigma$ if the pattern and the string are the same length and each character of the pattern is the same as the character of the string at that position, or the pattern has an $\ast$ at that position, \ie{} $\ast$ `matches' any symbol from the alphabet.  The fundamental problem is then, given two sets, $G$ and $B$, of strings over $\Sigma$, can we find a set of patterns of size at most $k$ such that every string in $G$ matches one of our patterns, and none of the strings in $B$ match any of our patterns.

\subsection{Separating Healthy Patterns from Diseased}
\label{sec:motivation_SNP}

A significant portion of bioinformatics and computational medicine efforts are focused on developing diagnostic tools. The identification of explanatory genes, uncovering of biomarkers, metabolic network analysis and protein interaction analysis all have as a key (but not sole) motivation the identification of differential markers of disease and consequently routes to treatment. Consider the following problem as a motivating archetypal example: we have two sets of individuals, healthy  and diseased and for each example we are given a string that encodes the single-nucleotide polymorphism (SNPs) states across the two copies of each genome, giving us two sets of strings $G$ and $B$\footnote{Whether healthy is $G$ and diseased is $B$ or vice versa depends on what information we wish the set of patterns to extract.}. A SNP has several alleles of which an individual has two. The individual may thus be homozygous in any of the alleles, or heterozygous with any choice of pairs of alleles, giving the underlying alphabet $\Sigma$.

It is easy to see that if we can identify patterns of SNPs that separate the healthy from the diseased individuals, we have a source of genetic information that may assist in explaining and treating the disease.

This problem is even more apparent in its computational form when considering a biologically motivated form of computation, \emph{i.e.}, evolutionary algorithms.

\subsection{Patterns in Multiparent Recombination}
\label{sec:motivation_GA}

The central mechanism for optimization in Genetic Algorithms (GAs) is the recombination of parent solutions to produce a new child solution which ideally retains the positive aspects of the parents. The mechanism derives from an analogy with sexual reproduction in biological evolution and hence typically combines two existing solutions to produce the offspring. In the optimization setting however, there's no conceptual reason for this restriction. Given that recombination can be viewed as a local move in the search space from one individual solution to another as mediated by a third individual solution, a natural generalization of this is to employ multiple parents in the hope of further refining the components of the solution that promote quality, while producing new solutions that effectively cover the search space.

The central theoretical formalization for describing this process is that of \emph{schemata}\footnote{We use the third declension neuter form of \emph{schema}, as it matches better the Greek roots of the word.}. An individual solution in a (simple) GA is described by an array, which we can represent as a string, of length $n$ over a given alphabet $\Sigma$. A \emph{schema} is a string of length $n$ over the same alphabet augmented with the special ``wild card'' character $\ast$, \ie{}, a pattern. A schema can then be thought of as representing a portion of the search space. The preservation of desirable shared characteristics of two or more parent individuals can then be viewed as the problem of defining a suitable schema. We can define a set $G$ using the individuals selected as parents for a recombination operation, and, if desired, a set $B$ from any individuals whose characteristics we may wish to avoid. The child individual(s) can then be generated from this schema with the wild cards replaced in whichever manner is chosen. Thus we can use schemata to model the basic operation of genetic recombination operators.

This idea not only models multiparent recombination but also multi-child recombination. When considering simply a set of parents from which we wish to generate a set of children, constructing schemata that are compatible with the parents is straightforward. A single schema that is a string of $n$ many $\ast$ symbols would suffice as a trivial solution and the natural solution where for each position, if all the parents agree on the same symbol, the schema has that symbol and $\ast$ otherwise also provides a simple solution. However in these cases it is reasonably easy to see that the schemata generated can easily be under-specified, leading to a loss of useful information, rendering the recombination operation ineffective. One solution to this problem is to ask for a small set of schemata that are compatible with the parents, but are incompatible with a set of forbidden strings -- akin to the list of forbidden elements in Tabu search. In this paper, we elaborate upon and examine this idea.

Some further complexity issues surrounding multiparent recombination have been examined in~\cite{CottaMoscato2005}.

\subsection{Our Contribution}
\label{sec:ccc}

In this paper we formalize the problem of finding a small set of patterns that match a set of strings, without matching a set of forbidden strings, as discussed in the introduction and examine its complexity. We call the general form of the problem \PI{} and introduce some useful variants. In most cases this problems turn out to be hard. We naturally then consider the problem from a Parameterized Complexity perspective. The problem has a rich parameter ecology and also provides an interesting example of a non-graph theoretic problem. Unfortunately for many parameterizations the problem turns out to be hard in this setting as well. The natural parameterization by the number of desired schemata is \Wtwo{}-hard. Even if we take the length of the strings as the parameter, the problem is \paraNP{}-complete. Table~\ref{tab:Results} gives a summary of the parameterized results, and some key open problems. It is also inapproximable and for some cases we obtain parameterized inapproximability results as well. The only case for which we are able to obtain fixed-parameter tractability relies on a small number of input strings which have a limited number of symbols which are different from a given ``base'' symbol.

\begin{table}[t]
  \centering
  \begin{tabular}{@{}llr@{}} \toprule
   Parameter & Complexity & Theorem\\ \midrule
   $k + \Card{\Sigma} + \Card{B}$       & \Wtwo{}-hard & \ref{thm:PI_W[2]-c}\\
   $k + \Card{\Sigma} + s + \Card{B}$ & \Wtwo{}-complete & \ref{cor:PIp_W[2]-complete}\\
   $n + d + \Card{B}$ & \paraNP{}-complete & \ref{cor:PIS_paraNP-c_short_strings}\\
   $\Card{\Sigma} + d + r + \Card{B}$ & \paraNP{}-complete & \ref{cor:PIP_PIPS_paraNP-c}\\
   $\Card{\Sigma} + d + s + \Card{B}$ & \paraNP{}-complete & \ref{cor:PIP_PIPS_paraNP-c}\\
   $d + \Card{G} + \Card{B}$ & \FPT{} & \ref{PI_FPT_few_small_strings}\\
   $\Card{\Sigma} + n$ & \FPT{} & \ref{thm:PI_FPT_k_Sigma_n}\\
   $k + n$ & \FPT{} & \ref{thm:PI_FPT_k_n}\\
   $\Card{G} + n$ & \FPT{} & \ref{thm:PI_FPT_G_n}\\
   $k + \Card{\Sigma} + d + r + \Card{B}$ & \FPT{} & \ref{thm:PIP_FPT_k_r_Sigma_B}\\
   $k + \Card{G} + \Card{B}$ & Open &\\
   $\Card{\Sigma} + \Card{G} + \Card{B}$ & Open & \\
   $k + \Card{\Sigma} + d$ & Open &\\
   \bottomrule
  \end{tabular}
  \caption{Summary of the parameterized results of the paper. $\Card{\Sigma}$ is the size of the alphabet, $n$ is the length of the strings and patterns, $\Card{G}$ and $\Card{B}$ are the sizes of the two input string sets, $k$ is the number of patterns, $r$ is the maximum number of $\ast$ symbols in a pattern, $s$ is the maximum number of non-$\ast$ symbols in a pattern and $d$ is the number of `non-base' elements in each string. Of course the usual inferences apply: tractable cases remain tractable when expanding the parameter and intractable cases remain intractable when restricting the parameter. We note that a number of cases remain open, of which we include some of the more pertinent here, however given the number of parameters under consideration, we refer the reader to Sections~\ref{sec:discussion_parameters} and~\ref{sec:conclusion} for a proper discussion of the open cases.} 
\label{tab:Results}
\end{table}

\subsection{Related Work}
\label{sec:related}

The identification of patterns describing a set of strings forms a well studied family of problems with a wide series of applications. Although, as best as we can determine, the precise problems we studied here have not yet been considered, a number of interesting related problems are explored in the literature. We present here a selection of some of the more relevant and interesting results, however these can at best form a basis for further exploration by the interested reader.

One of most immediately similar variants is that where pattern variables are allowed. In contrast to the work here, these variables can act as substrings of arbitrary length. Keans and Pitt~\cite{KearnsPitt1991} give a family of polynomial time algorithms for learning the language generated by a single such pattern with a given number $k$ of pattern variables. Angluin~\cite{Angluin1980} studies the inverse problem of generating a pattern, with a polynomial time algorithm for the case where the pattern contains a single pattern variable being the central result. We note that a central difference here is the repeated use of variables, allowing the same undefined substring to be repeated. The properties of these pattern languages have since been studied in some detail, far beyond the scope of this paper.

Bredereck, Nichterlein and Niedermeier~\cite{BredereckNichterleinNiedermeier2013} employ a similar, but not identical, formalism to that employed here, but study the problem of taking a set of strings and a set of patterns and determining whether the set of strings can be altered to match the set of patterns. In their formalism patterns are strings over the set $\{\Box,\star\}$. We note in particular though that their definition of matching differs from our definition of compatibility in that a string matches a pattern if and only if the string has the special symbol $\star$ exactly where the pattern does. They show this problem to be \NP{}-hard, but in \FPT{} when parameterized by the combined parameter of the number of patterns and the number of strings. They also present an ILP based implementation and computational results. Bredereck \emph{et al.}~\cite{Brederecketal2015} examine forming teams, \emph{i.e.}, mapping the set of strings to the set of patterns in a consistent manner. They use a similar basis, excepting that the special $\star$ symbol in a pattern now matches any symbol in a string and that the $\Box$ symbol requires homogeneity of the matched strings (\emph{i.e.} the symbol it matches is not specified, but all matching strings must have the same symbol at that point). They give a series of classification results, with the problem mostly being intractable, but in \FPT{} for the number of input strings, the number of \emph{different} input strings and the combined parameter of alphabet size with the length of the strings.

Gramm, Guo and Niedermeier~\cite{Gramm2006} study another similar problem, \textsc{Distinguishing Substring Selection}, where the input is two sets of strings (``good'' and ``bad''), and two integers $d_{g}$ and $d_{b}$ with the goal of finding a single string of length $L$ whose Hamming distance from all length $L$ substrings of every ``good'' string is at least $d_{g}$ and from at least one length $L$ substring for each ``bad'' string is at most $d_{b}$. An extension of the \textsc{Closest String}~\cite{GrammNR03,Li2002} and \textsc{Closest Substring}~\cite{FellowsGN2002} problems, the problem has a ptas~\cite{Deng2002} but they show that it is \Wone{}-hard when parameterized by any combination of the parameters $d_{g}$, $d_{b}$ and the number of ``good'' or ``bad'' strings. Under sufficient restriction they demonstrate an \FPT{} result, requiring a binary alphabet, a `dual' parameter $d'_{g} = L - d_{g}$ and that $d'_{g}$ is optimal in the sense that it is the minimum possible value. We note that, in relation to the problems studied here, although the number of $\ast$ symbols in the patterns provides an upper-bound for the Hamming distance, the Hamming distance for a set of strings may be much lower; consider a set of strings with one position set to $1$ and all others to $0$ such that for every possible position there is a string with a $1$ at that point, then the string (or indeed substring) of all $0$ has Hamming distance at most one from each input string, but a single pattern would need to be entirely $\ast$ symbols to match the entire set.

Hermelin and Rozenberg introduce a further variant of the \textsc{Closest String} problem~\cite{HermelinR15}, the \textsc{Closest String with Wildcards} problem. The input is a set of strings $\{s_{i}\}$, which may include wildcard characters, and an integer $d$. The goal is to find a string with hamming distance at most $d$ to each $s_{i}$. The solution is required to have no wildcard characters. The examine a number of parameters: the length $n$ of the input strings, the number $m$ of input strings, $d$, the number $\Card{\Sigma}$ of characters in the alphabet, and the minimum number $k$ of wildcard characters in any input string. They show that the problem is in \FPT{} (with varying explicit running times) when parameterized by $m$, $m+n$, $\Card{\Sigma} + k + d$ and $k+d$. They also show that the special case where $d=1$ can be solved in polynomial time, whereas the problem is \NP{}-hard for every $d \geq 2$.

Bulteau \emph{et al.}~\cite{BulteauHKN2014} also give a survey of the parameterized complexity of a variety of more distantly related string problems, with similar multivariate parameterizations as in other work in this area. They cover, amongst others, \textsc{Closest String}, \textsc{Closest Substring}, \textsc{Longest Common Subsequence}, \textsc{Shortest Common Supersequence}, \textsc{Shortest Common Superstring}, \textsc{Multiple Sequence Alignment} and \textsc{Minimum Common String}.

Introduced by Cannon and Cowen~\cite{CannonC04}, the \textsc{Class Cover} problem is a geometric relative of \PI{} where the input is two sets of points colored red and blue, with the goal of selecting a minimum set of blue points (centers) that ``cover'' the full set of blue points, in the sense that any blue point is closer to its nearest center than any red point. It is \NP{}-hard with an $O(\log n + 1)$-factor approximation algorithm, bearing a close similarity to \DS{}.

\section{Preliminaries and Definitions}
\label{sec:prelims}

\sloppypar We now give the relevant definitions for the complexity analysis that follows. In the reductions we use the well known \DS{} and \VC{} problems. The graphs taken as input for these problems are simple, undirected and unweighted. To assist with notation and indexing, we take the vertex set $V(\mathfrak{G})$ of a graph $\mathfrak{G}$ to be the set $\{1,\ldots ,n\}$. The edge set $E(\mathfrak{G})$ is then a set of pairs drawn from $V(\mathfrak{G})$ and we denote the edge between vertices $i$ and $j$ by $ij$ ($=ji$). The \SetCover{} and \kFS{} problems are also employed. The problems are defined as follows:

\problem{\DS{}}{A graph $\mathfrak{G}$ and an integer $k$.}{Is there a set $V' \subseteq V(\mathfrak{G})$ with $\Card{V'} \leq k$ such that for every $u \in V(\mathfrak{G})$ there exists a $v \in V'$ with $u \in N(v)$?}

\problem{\VC{}}{A graph $\mathfrak{G}$ and an integer $k$.}{Is there a set $V' \subseteq V(\mathfrak{G})$ with $\Card{V'} \leq k$ such that for every $uv \in E(\mathfrak{G})$ we have $u \in V'$ or $v \in V'$?}

\problem{\SetCover{}}{A base set $U$, a set $S \subseteq \mathcal{P}(U)$ and an integer $k$.}{Is there a set $S' \subseteq S$ with $|S'| \leq k$ such that $\bigcup S' = U$?}

\problem{\kFS{}}{An $n\times m$ $0\text{-}1$ matrix $M$, an $n\times 1$ $0\text{-}1$ vector $f$ and an integer $k$.}{Is there a set of indices $I \subseteq \{1,\ldots, m\}$ with $\Card{I} \leq k$ such that for all $a,b$ where $f_{a} \neq f_{b}$ there exist $i \in I$ such that $M_{a,i} \neq M_{b,i}$?}

We note the following key classification results:

\begin{itemize}
\item \sloppypar \DS{} is \NP{}-complete, $O(\log n)\text{-}\APX{}$-hard\footnote{That is there exists some $c > 0$ such that \DS{} has no $c\cdot\log n$-factor approximation algorithm unless $\P{} = \NP{}$.} and \Wtwo{}-complete when parameterized by $k$, the size of the dominating set.
\item \sloppypar \VC{} is \NP{}-complete and \APX{}-hard, and remains \NP{}-complete when the input is a planar graph~\cite{GareyJohnson79}.
\item \sloppypar \SetCover{} is \Wtwo{}-complete when parameterized by the size of the set cover.
\item \sloppypar \kFS{} is \Wtwo{}-complete when parameterized by the size of the feature set~\cite{CottaMoscato2002}.
\end{itemize}

We also employ a parameterized version of the \textsc{Model Checking} problem, which takes as input a finite structure and a logical formula and asks the question of whether the structure is a model of the formula, \emph{i.e.} whether there is a suitable assignment of elements of the universe of the structure to variables of the formula such that the formula evaluates to true under that assignment. The parameter is the length of the logic formula. While we informally introduce the finite structural elements as needed, we briefly describe here the fragments of first-order logic we employ. Let $\Sigma_{0} = \Pi_{0}$ be the set of unquantified Boolean formulae. The classes $\Sigma_{t}$ and $\Pi_{t}$ for $t > 0$ can be defined recursively as follows:
\begin{align*}
\Sigma_{t} &= \{\exists x_{1}\ldots \exists x_{k}\varphi \mid \varphi \in \Pi_{t-1}\}\\
\Pi_{t} &= \{\forall x_{1} \ldots \forall x_{k} \varphi \mid \varphi \in \Sigma_{t-1}\}
\end{align*}
The class $\Sigma_{t,u}$ is the subclass of $\Sigma_{t}$ where each quantifier block after the first existential block has length at most $u$. We note that trivially $\Pi_{t-1} \subset \Sigma_{t}$. We note that these classes are specified in prenex normal form, and are, in general, not robust against Boolean combinations of formulae. In general, the process of converting a formula to prenex normal form (where all quantifiers are ``out the front'') increases the number of quantifier alternations. An analog of the $\Sigma$ classes is $\Sigma^{\ast}_{t,u}$. Let $\Theta_{0,u}$ be the set of quantifier free formulae, and $\Theta_{t,u}$ for $t > 0$ be the set of Boolean combinations of formulae where each leading quantifier block is existential and quantifies over a formula in $\Theta_{t-1,u}$, where the length of each quantifier block is at most $u$. That is, the formulae in $\Theta_{t,u}$ are not required to be in prenex normal form, and Boolean connectives may precede some quantifiers. We can deal with leading universal quantification by the normal expedient of the introduction of a trivial existential block. Then $\Sigma^{\ast}_{t,u}$ is the class of formulae of the form $\exists x_{1} \ldots \exists x_{k} \varphi$ where $\varphi \in \Theta_{t-1,u}$ and where $k$ may be greater than $u$. 

Thus we refer to the \textsc{Model Checking} problem as \textsc{MC}($\Phi$) where $\Phi$ is the first-order fragment employed. In the parameterized setting, \MCSigmatu{} is \Wt{}-complete for every $u \geq 1$, and \MCSigmaStartu{} is \Wstart{}-complete for every $u \geq 1$. The $\mathsf{W}^{\ast}$-hierarchy is the hierarchy analogous to the $\mathsf{W}$-hierarchy obtained from using \MCSigmaStartu{} as the complete problem instead of \MCSigmatu{}. While it is known that $\Wone{} = \Wstarone{}$ and $\Wtwo{} = \Wstartwo{}$; for $t \geq 3$, the best known containment relationship is $\Wt \subseteq \Wstart{} \subseteq \Wtwotminustwo{}$. For more detail on these results and the full definitions relating to first-order logic and structures we refer the reader to~\cite{FlumGrohe2006}. The $\mathsf{W}^{\ast}$-hierarchy, introduced by Downey, Fellows and Taylor~\cite{DowneyFT96} but more fully explored later~\cite{ChenFlumGrohe2007,FlumGrohe2006} is a parameterized hierarchy which takes into account the Boolean combinations of quantified first-order formulae, but is otherwise similar to the more usual $\mathsf{W}$-hierarchy.

In several of our intractability results we make use of the class \paraNP{}, and a useful corollary due to Flum and Grohe with a detailed explanation in~\cite{FlumGrohe2006} (presented as Corollary 2.16). The class \paraNP{} is the direct parameterized complexity translation of \NP{}, where we replace ``polynomial-time'' with ``fixed-parameter tractable time'' (or fpt-time in short) in the definition. Flum and Grohe's result states that if, given a parameterized problem $(\Pi,\kappa)$, the classical version of the problem $\Pi$ is \NP{}-complete for at least one fixed value of $\kappa$, then $(\Pi, \kappa)$ is \paraNP{}-complete. As may be expected, $\FPT{} = \paraNP{}$ if and only if $\P{} = \NP{}$, thus \paraNP{}-completeness is strong evidence of intractability. We also make reference to parameterized approximation. A parameterized approximation algorithm is, in essence, a standard approximation algorithm, but where we relax the running time to fpt-time, rather than polynomial-time. We refer to~\cite{Marx2008} for a full introduction to this area.

The other parameterized complexity theory employed is more standard, thus for general definitions we refer the reader to standard texts~\cite{DowneyFellows2013,FlumGrohe2006}.

We write $A \leq_{FPT} B$ to denote that there exists a parameterized reduction from problem $A$ to problem $B$, and similarly $A \leq_{P} B$ to denote the existence of a polynomial-time many-one reduction from problem $A$ to problem $B$. We also use \emph{strict polynomial-time reductions} to obtain some approximation results. A strict reduction is one that, given two problems $A$ and $B$, guarantees that the approximation ratio for $A$ is at least as good as that of $B$. In the cases we present, we employ them for approximation hardness results, so the precise ratio is not discussed. For a full definition of strict reductions (and other approximation preserving reductions) we refer to~\cite{Crescenzi1997}.

\begin{definition}[Pattern]\label{def:pattern}
  A \emph{pattern} is a string over an alphabet $\Sigma$ and a special symbol $\ast$.
\end{definition}

Given a string $s \in \Sigma^{\ast}$ and an integer $i$, we denote the $i$th symbol of $s$ by $s[i]$.

\begin{definition}[Compatible]\label{def:compatible}
A pattern $p$ is compatible with a string $g$, denoted $p \to g$, if for all $i$ such that $p[i] \neq \ast$ we have $g[i] = p[i]$. If a pattern and string are not compatible, we write $p \not\to g$. We extend this notation to sets of strings, writing $p \to G$ to denote $\forall g \in G, p \to g$ and $P \to G$ for $\forall g \in G \exists p \in P, p \to g$.
\end{definition}

\begin{definition}[$G$-$B$-Separated Sets]\label{def:separated_sets}
A set $P$ of patterns \emph{$G$-$B$-separates} an ordered pair $(G,B)$ of sets of strings, written $P \to (G,B)$ if
\begin{itemize}
\item $P \to G$, and
\item for every $b \in B$ and $p \in P$ we have $p \not\to b$.
\end{itemize}
\end{definition}

Thus we can state the central problem for this paper:

\problem{\PI{}}{A finite alphabet $\Sigma$, two disjoint sets $G, B \subseteq \Sigma^{n}$  of strings and an integer $k$.}{Is there a set $P$ of patterns such that $\Card{P} \leq k$ and $P \to (G,B)$?}

The complexity analysis of the problem in the parameterized setting leads to the definition of a second, subsidiary problem which allows a convenient examination of sets of strings which are very similar.

\begin{definition}[Small String]\label{def:small_string}
  A string $s$ over an alphabet $\Sigma$ is $d$-\emph{small} if, given an identified symbol $\sigma \in \Sigma$, for exactly $d$ values of $i$, $p[i] \neq \sigma$.

We call $\sigma$ the \emph{base} symbol.

A set of strings $S$ is $d$-small if, given a fixed base symbol all strings in $S$ are $d$-small.
\end{definition}

This restriction on the structure of the input gives further insight into the complexity of \PI{} and is key to some of the tractability results in Section~\ref{sec:easy}. For convenience we phrase a restricted version of the \PI{} problem:

\problem{\PIS{}}{An alphabet $\Sigma$, two disjoint $d$-small sets $G, B \subseteq \Sigma^{n}$, an integer $k$.}{Is there a set $P$ of patterns with $\Card{P} \leq k$ such that $P \to (G,B)$?}

From the perspective of multiparent recombination, minimizing the number of wildcard symbols in each pattern is also an interesting objective:

\problem{\PIP{}}{An alphabet $\Sigma$, two disjoint sets $G, B \subseteq \Sigma^{n}$, integers $k$ and $r$.}{Is there a set $P$ of patterns with $\Card{P} \leq k$ such that $P \to (G,B)$ and for each $p \in P$ the number of $\ast$ symbols in $p$ is at most $r$?}

We implicitly define the obvious intersection of the two restricted problems, \PIPS{}.

From a combinatorial perspective, the inverse problem is also interesting:

\problem{\PIp{}}{An alphabet $\Sigma$, two disjoint sets $G, B \subseteq \Sigma^{n}$, integers $k$ and $s$.}{Is there a set $P$ of patterns with $\Card{P} \leq k$ such that $P \to (G,B)$ and for each $p \in P$ the number of non-$\ast$ symbols in $p$ is at most $s$?}

\section{Hard Cases of the Pattern Identification Problem}
\label{sec:hard}

We first examine the intractable cases of the \PI{} problem. This narrows down the source of the combinatorial complexity of the problem.

\begin{theorem}\label{thm:PI_W[2]-c}
 \sloppypar \PI{} is \Wtwo{}-hard when parameterized by $k$, even if $\Card{\Sigma} = 2$ and $\Card{B} = 1$.
\end{theorem}

\begin{lemma}\label{lemma:DS<=PI}
  $\DS{} \leq_{FPT} \PI{}$.
\end{lemma}

\begin{figure}[htb]
  \centering
  \subfloat{\raisebox{-0.5\height}{\begin{tikzpicture}[shorten >=1pt,node distance=2cm,auto] 
    \node[draw,circle,fill=red!40] (q_0)                      {$1$}; 
    \node[draw,circle] (q_1) [above right of=q_0] {$2$}; 
    \node[draw,circle] (q_2) [above left of=q_0]  {$3$};
    \node[draw,circle] (q_3) [below right of=q_0] {$4$};
    \node[draw,circle] (q_4) [below left of=q_0]  {$5$};

    \path
    (q_0) edge (q_1) edge (q_2) edge (q_3) edge (q_4)
    (q_1) edge (q_2)
    (q_3) edge (q_4);
  \end{tikzpicture}}}\hspace{48pt}$\rightarrow$\hspace{48pt}
  \subfloat{
    \begin{tabular}{c|ccccc}
      \multirow{5}{*}{$G$} & 1 & 1 & 1 & 1 & 1\\
      &1&1&1&0&0\\
      &1&1&1&0&0\\
      &1&0&0&1&1\\
      &1&0&0&1&1\\
      \midrule
      $B$ & 0 & 0 & 0 & 0 & 0
    \end{tabular}
  }\\
  \subfloat{
    $P = \{1\ast\ast\ast\ast\}$
  }
  \caption{An example of the reduction used in Lemma~\ref{lemma:DS<=PI} with $k=1$. The dominating set is highlighted in red, and the correspond set of patterns (a singleton) is shown.}
  \label{fig:DS_red_example}
\end{figure}
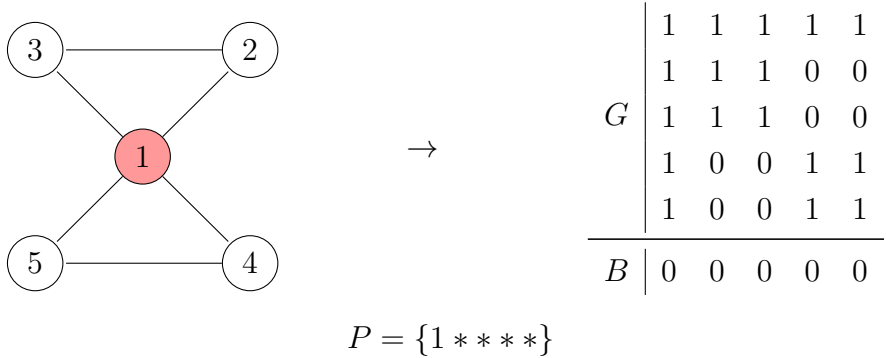

\begin{proof}
Let $(\mathfrak{G},k)$ be an instance of \DS{}. Let $n = \Card{V(\mathfrak{G})}$ and assume $V(\mathfrak{G}) = \{1,\ldots, n\}$. We construct an instance $(\Sigma, G, B, k)$ as follows:
\begin{enumerate}
\item $\Sigma = \{1, 0\}$,
\item $G = \{g_{1},\ldots, g_{n}\}$ where for each $i$, $g_{i} \in \Sigma^{n}$ where for every $j$, $g_{i}[j] = 1$ if $ij \in E(\mathfrak{G})$ or $i = j$ and $g_{i}[j] = 0$ otherwise, 
\item $B = \{0^{n} \}$.
\end{enumerate}

An example of the reduction is given in Figure~\ref{fig:DS_red_example}.

\begin{claim}
  If $(G,k)$ is a \Yes{} instance of \DS{} then $(\Sigma, G, B, k)$ is a \Yes{} instance of \PI{}.
\end{claim}

Let $D \subseteq V(\mathfrak{G})$ with $\Card{D} \leq k$ be a dominating set witnessing that $(\mathfrak{G},k)$ is a \Yes{} instance. We can construct a witness set of patterns $P$ with $\Card{P} = \Card{D}$ such that $P \to (G,B)$. For each $i \in D$, we create a pattern $p_{i}$ where $p_{i}[i] = 1$ and $p_{i}[j] = \ast$ for all $j \neq i$.

As $D$ is a dominating set, for every vertex $j \in V(\mathfrak{G})$ there is a vertex $i \in D$ such that $ij \in E(\mathfrak{G})$. Then for string $g_{j} \in G$, pattern $p_{i}$ is compatible with $g_{j}$ as by construction $g_{j}[i] = 1$. Therefore for every string $g \in G$ there exists $p \in P$ such that $p \to g$. 

Moreover there is no $p \in P$ such that $p \to b$ where $b \in B$. As $B$ consists of the single string $b = 0^{n}$ and for each $p \in P$ exactly one element is neither $0$ nor $\ast$ there is one position where the pattern does not match $b$.

Thus $(\Sigma, G, B, k)$ is a \Yes{} instance.

\begin{claim}
  If $(\Sigma, G, B, k)$ is a \Yes{} instance of \PI{} then $(\mathfrak{G},k)$ is a \Yes{} instance of \DS{}
\end{claim}

Let $P$ with $\Card{P} \leq k$ be the witness set of patterns that $(\Sigma, G, B, k)$ is a \Yes{} instance. Inductively, we may assume that every $p \in P$ is compatible with at least one element of $G$, if not, $P \setminus \{p\}$ constitutes an alternate witness.

First we note that no $p \in P$ can consist of only $\ast$ and $0$ symbols, as this would be compatible with the single element of $B$. Therefore each $p \in P$ has at least one $i$ such that $p[i] = 1$.

Consider a $p \in P$, and the corresponding set of vertices $V_{p}$ (i.e. each position $i$ where $p[i] = 1$). Let $g_{j} \in G$ be a string such that $p \to g_{j}$. By construction for every $i \in V_{p}$, $ij \in E(\mathfrak{G})$. Let $V_{p\to g}$ be the set of vertices corresponding to the set $G_{p} \subseteq G$ where $p \to G_{p}$. Each vertex $j \in V_{p \to g}$ is adjacent (or identical) to every vertex in $V_{p}$. Thus we may select arbitrarily a single vertex from $V_{p}$ to be in the dominating set $D$.

Thus we have $D$ with $\Card{D} \leq \Card{P} \leq k$, where every vertex in $V(\mathfrak{G})$ is adjacent (or identical) to some vertex in $D$. Therefore $(\mathfrak{G},k)$ is a \Yes{} instance.

The construction can be clearly performed in fpt time (in fact, polynomial time), and the lemma follows.
\end{proof}

\begin{proof}[Proof of Theorem~\ref{thm:PI_W[2]-c}]
  The theorem follows immediately from Lemma~\ref{lemma:DS<=PI}.
\end{proof}

The structure of the reduction then gives the following:

\begin{corollary}\label{cor:PIp_W[2]-complete}
\sloppypar \PIp{} is \Wtwo{}-complete when parameterized by $k$ and \NP{}-complete even when $s = 1$, $\Card{\Sigma} = 2$ and $\Card{B} = 1$.
\end{corollary}

\begin{proof}
The \Wtwo{}-hardness is apparent from the proof of Lemma~\ref{lemma:DS<=PI} (in fact the restriction would make the proof simpler). To show containment in \Wtwo{}, we reduce \PIp{} to \MCSigmaTwoOne{}. The first-order structure is equipped with four unary relations $N$, $\Sigma$, $G$ and $B$ and a binary function symbol $C$. Each string is represented by an integer, according to an arbitrary fixed ordering. $Gi$ is true if string $i$ is in $G$, $Bi$ is true if string $i$ is in $B$. $\Sigma\sigma$ is true if $\sigma \in \Sigma$ and $Ni$ is true if $i \in \mathbb{N}$. The function $C:\mathbb{N} \times \mathbb{N} \to \Sigma$ is defined $Cij = \sigma$ if $\sigma$ is the $j$th symbol of string $i$.

We now provide the first-order formula expressing \PIp{}:
\begin{align*}
\exists i_{1,1},\ldots,i_{k,s},c_{1,1},\ldots,c_{k,s}\forall j ((\bigwedge_{l\in[k],b\in[s]} Ni_{l,b}) \wedge\\ (\bigwedge_{l\in[k],b\in[s]}\Sigma c_{l,b}) \wedge\\ (Gj \to (\bigvee_{l\in[k]}(\bigwedge_{b\in[s]}Cji_{l,b} = c_{l,b}))) \wedge\\ (Bj \to (\bigwedge_{l\in[k]}(\bigvee_{b\in[s]}Cji_{l,b} \neq c_{l,b}))))
\end{align*}

The formula states that a solution to \PIp{} consists of $k$ sets of $s$ symbols along with positions such that for each string in $G$, for at least one set of symbols, the string is compatible and for each string in $B$ no set of symbols is compatible.

Containment in \NP{} can be demonstrated by the usual polynomial verification approach (indeed in much the same format as the above formula).
\end{proof}

\begin{corollary}\label{cor:no_fpt_approx}
\sloppypar  \PI{} has no constant factor fpt-approximation algorithm unless $\FPT{} = \Wtwo{}$ and there exists a $c \geq 0$ such that \PI{} has no $c\cdot\log n$ polynomial time approximation algorithm unless $\P{} = \NP{}$, even when $\Card{\Sigma} = 2$ and the optimization goal is $\min k$.
\end{corollary}

\begin{proof}
  As \DS{} has no constant factor fpt-approximation~\cite{ChenLin2015} unless $\FPT{} = \Wtwo{}$ and no $c\cdot\log n$ polynomial time approximation~\cite{RazSafra1997} for some $c > 0$ unless $\P{} = \NP{}$ and the reduction of Lemma~\ref{lemma:DS<=PI} is a strict polynomial-time reduction, the corollary follows.
\end{proof}

Given the construction in the proof of Lemma~\ref{lemma:DS<=PI}, we can deduce that one source of complexity might be the freedom (unboundedness) in the alphabet and the structure of the strings. We demonstrate that restricting these parameters is fruitless from a computational complexity perspective.

\begin{corollary}\label{cor:PI_NPc_small_strings_small_alphabet}
\PIS{} is \NP{} complete even when $\Card{\Sigma} = 2$, $d = 4$, $s = 1$ and $\Card{B} = 1$. 
\end{corollary}

\begin{proof}
As \DS{} is \NP{}-complete on planar graphs of maximum degree 3~\cite{GareyJohnson79}, the number of $1$s in each string in the construction of the proof of Lemma~\ref{lemma:DS<=PI} is at most 4, where we take the base symbol to be $0$.
\end{proof}

This result also demonstrates the following:

\begin{lemma}\label{lemma:PIP_PIPS_NPc}
\PIPS{} and\\ \PIP{} are both \NP{}-complete even when $\Card{\Sigma} = 2$, $d = 4$, $r = 9$ and $\Card{B} = 1$.
\end{lemma}

\begin{proof}
Following Corollary~\ref{cor:PI_NPc_small_strings_small_alphabet}, we can see from the construction given in the proof of Lemma~\ref{lemma:DS<=PI} that for each $p \in P$, instead of setting $p[i] = \ast$ for each $i$ not in the dominating set, we can choose $r$ to be nine, and set $p[i] := 1$ if $i$ is in the dominating set, $p[j] = \ast$ for the at most three values of $j$ such that $ij \in E(\mathfrak{G})$ and the at most six additional values of $j$ at distance two\footnote{As $\mathfrak{G}$ has maximum degree three, each neighbor of $i$ has at most two other neighbors, so the patterns representing each of these neighbors has a $1$ in the $i$th position, a $1$ for its own position and two other $1$s. Therefore we need only three $\ast$ symbols for the neighbors themselves, and two more per neighbor for the distance two neighborhood.} from $i$, and $p[l] = 0$ for all other $l \in \{1,\ldots,n\}$. For the reverse argument, we have similar conditions as before, at least one symbol of each pattern must be a $1$ and at most four can be $1$s. With at most nine $\ast$ symbols, the pattern is compatible with all the strings that the corresponding vertex dominates, and all other symbols in these strings are $0$.
\end{proof}

\begin{corollary}\label{cor:PIP_PIPS_paraNP-c}
The following are true:
  \begin{enumerate}
  \item \PIPS{} is \paraNP{}-complete when parameterized by $\Card{\Sigma} + d + r + \Card{B}$.
  \item \PIP{} is \paraNP{}-complete when parameterized by $\Card{\Sigma} + r + \Card{B}$.
  \item \PIps{} is \paraNP{}-complete when parameterized by $\Card{\Sigma} + d + s + \Card{B}$.
  \item \PIp{} is \paraNP{}-complete when parameterized by $\Card{\Sigma} + s + \Card{B}$.
  \end{enumerate}
\end{corollary}

\begin{proof}
The result are obtained as follows:

\begin{enumerate}
\item Lemma~\ref{lemma:PIP_PIPS_NPc} gives \NP{}-completeness with fixed $\Card{\Sigma}$, $d$, $r$ and $\Card{B}$. With Corollary 2.16 from~\cite{FlumGrohe2006}, the result follows.
\item The preservation of hardness when taking subsets of a set of parameters gives the result from $1$.
\item Corollary~\ref{cor:PI_NPc_small_strings_small_alphabet} shows \NP{}-completeness with fixed $\Card{\Sigma}$, $d$, $s$ and $\Card{B}$. Corollary 2.16 from~\cite{FlumGrohe2006} completes the result.
\item The result follows immediately from $3$.
\end{enumerate}
\end{proof}

\sloppypar We note that \DS{} is in \FPT{} for graphs of bounded degree, so we do not obtain a \Wtwo{}-hardness result. However we can tighten this result a little further:

\begin{theorem}\label{thm:PI_NP-c_small_strings}
  \PI{} is \NP{}-complete and \APX{}-hard even when $\Sigma = \{0,1\}$ and all strings have at most two symbols as $1$ (equiv. at most two symbols as $0$) and $\Card{B} = 1$.
\end{theorem}

\begin{lemma}\label{lem:VC<=PI}
  $\VC{} \leq_{P} \PI{}$.
\end{lemma}

\begin{figure}[htb]
  \centering
  \subfloat{\raisebox{-0.5\height}{\begin{tikzpicture}[shorten >=1pt,node distance=2cm,auto] 
    \node[draw,circle,fill=red!40] (q_0)                      {$1$}; 
    \node[draw,circle,fill=red!40] (q_1) [above right of=q_0] {$2$}; 
    \node[draw,circle] (q_2) [above left of=q_0]  {$3$};
    \node[draw,circle,fill=red!40] (q_3) [below right of=q_0] {$4$};
    \node[draw,circle] (q_4) [below left of=q_0]  {$5$};

    \path
    (q_0) edge (q_1) edge (q_2) edge (q_3) edge (q_4)
    (q_1) edge (q_2)
    (q_3) edge (q_4);
  \end{tikzpicture}}}\hspace{48pt}$\rightarrow$\hspace{48pt}
  \subfloat{
    \begin{tabular}{c|ccccc}
      \multirow{6}{*}{$G$} & 1 & 1 & 0 & 0 & 0\\
      &1&0&1&0&0\\
      &1&0&0&1&0\\
      &1&0&0&0&1\\
      &0&1&1&0&0\\
      &0&0&0&1&1\\
      \midrule
      $B$ & 0 & 0 & 0 & 0 & 0
    \end{tabular}
  }\\
  \subfloat{
    $P = \{1\ast\ast\ast\ast,\;\;\; \ast 1\ast\ast\ast,\;\;\; \ast\ast\ast 1\ast\}$
  }
  \caption{An example of the reduction used in Lemma~\ref{lem:VC<=PI} with $k=3$. The vertex cover is highlighted in red, and the correspond set of patterns is shown.}
  \label{fig:VC_red_example}
\end{figure}
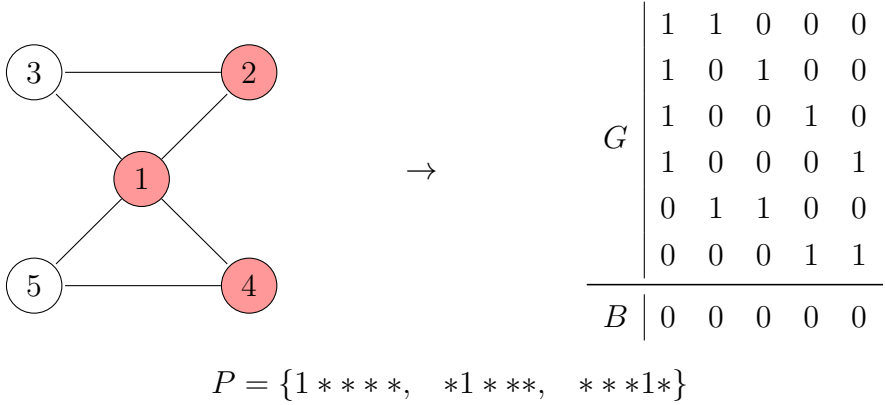

\begin{proof}
Given an instance $(\mathfrak{G},k)$ of \VC{} with $V(\mathfrak{G}) = \{1,\ldots, n\}$, we construct an instance $(\Sigma, G, B, k)$ of \PI{} as follows:
\begin{enumerate}
\item $\Sigma = \{0,1\}$.
\item $G = \{g_{ij} \mid ij \in E(\mathfrak{G})\}$ with $g_{ij} \in \Sigma^{n}$ where $g_{ij}[i] = g_{ij}[j] = 1$ and $g_{ij}[u] = 0$ for $u \neq i,j$.
\item $B = \{0^{n}\}$.
\end{enumerate}

Clearly this construction can be performed in polynomial time. The construction is illustrated in Figure~\ref{fig:VC_red_example}.

\begin{claim}
  If $(\mathfrak{G},k)$ is a \Yes{} instance of \VC{} then $(\Sigma,G,B,k)$ is a \Yes{} instance of \PI{}.
\end{claim}

Let $V' \subseteq V(\mathfrak{G})$ where $\Card{V'} \leq k$ be a vertex cover witnessing that $(\mathfrak{G},k)$ is a \Yes{} instance of \VC{}. We construct a set of patterns $P$ with $\Card{P} = \Card{V'}$ that is a solution for $(\Sigma, G, B, k)$ where for each $i \in V'$ there is a pattern $p_{i} \in P$ with $p_{i}[i] = 1$ and $p_{i}[j] = \ast$ for $j \neq i$. For each edge $ij \in E(\mathfrak{G})$, either $i \in V'$ or $j \in V'$ (or both). Therefore for the string $g_{ij}$ corresponding to $ij$, we have either $p_{i} \in P$ or $p_{j} \in P$ such that $p_{i} \to g_{ij}$ or $p_{j} \to g_{ij}$. Hence $P \to G$. Moreover there is no $p_{i} \in P$ such that $p_{i} \to b$ where $b$ is the single element of $B$ as each $p_{i}$, by construction, contains a $1$, whereas $b$ consists of only $0$s. Therefore $(\Sigma, G, B, k)$ is a \Yes{} instance of \PI{}.

\begin{claim}
  If $(\Sigma,G,B,k)$ is a \Yes{} instance of \PI{} then $(\mathfrak{G},k)$ is a \Yes{} instance of \VC{}.
\end{claim}

Let $P$ with $\Card{P} \leq k$ be the set of patterns witnessing the fact that $(\Sigma,G,B,k)$ is a \Yes{} instance of \PI{}. We may assume without loss of generality that for every $p \in P$, there exists some $g \in G$ such that $p \to g$. Each $p \in P$ must contain at least one $1$, otherwise $p \to b$ where $b$ is the single element of $B$. No $p \in P$ can contain more than two $1$s, as there exists $g \in G$ such that $p \to g$, and every such $g$ has exactly two $1$s. We note that if a pattern $p$ has two $1$s, then there is exactly one $g \in G$ such that $p \to g$.

Let $P_{1} \subseteq P$ be the set of patterns with exactly one $1$ and $P_{2} \subseteq P$ be the set of patterns with exactly two $1$s. We have $P_{1} \cup P_{2} = P$. We construct a vertex cover $V' \subseteq V(\mathfrak{G})$ with $\Card{V'} \leq \Card{P}$ as follows:
\begin{enumerate}
\item for each $p \in P_{1}$ add $i$ to $V'$ where $p[i] = 1$,
\item for each $p \in P_{2}$ where $p[i] = p[j] = 1$, arbitrarily add one of $i$ or $j$ to $V'$.
\end{enumerate}
Consider every edge $ij \in E(\mathfrak{G})$, then for the corresponding $g_{ij} \in G$ there exists a $p \in P$ such that $p \to g_{ij}$. As each $p$ has at least one $1$, this $1$ must be at position $i$ or $j$ (or both). Therefore $i$ or $j$ is in $V'$ (or perhaps both), therefore $V'$ forms a valid vertex cover for $\mathfrak{G}$.
\end{proof}

\begin{proof}[Proof of Theorem~\ref{thm:PI_NP-c_small_strings}]
\sloppypar The \NP{}-hardness follows from Lemma~\ref{lem:VC<=PI}. The containment in \NP{} follows from the usual verification algorithm. The \APX{}-hardness follows as the reduction of Lemma~\ref{lem:VC<=PI} is strict and \VC{} is \APX{}-hard~\cite{Hastad1997}.
\end{proof}



Finally, as restricting the alphabet did not reduce the complexity, we consider the case where the strings themselves are short. Again the problem is hard, but we note that to achieve this reduction we relax the bound on $\Sigma$ (or in Parameterized Complexity terms, $\Card{\Sigma}$ is no longer a parameter -- if $\Card{\Sigma}$ is a parameter, the problem is in \FPT{}).

\begin{theorem}\label{thm:PIS_NP-c_short_strings}
\sloppypar \PIS{} is \NP{}-complete even when $n = 4$, $d = 4$ and $\Card{B} = 1$.  
\end{theorem}

\begin{lemma}\label{lem:VC<=PIS_short_strings}
\sloppypar $\PVC{} \leq_{P} \PIS{}$ even when the length of strings is restricted to $4$. 
\end{lemma}

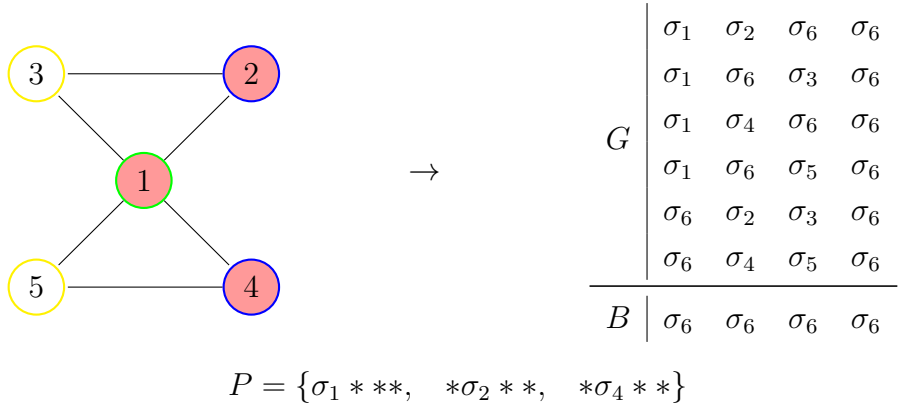
\begin{figure}[htb]
  \centering
  \subfloat{\raisebox{-0.5\height}{\begin{tikzpicture}[shorten >=1pt,node distance=2cm,auto] 
    \node[draw,green,circle,fill=red!40,text=black,thick] (q_0)                      {$1$}; 
    \node[draw,blue,circle,fill=red!40,text=black,thick] (q_1) [above right of=q_0] {$2$}; 
    \node[draw,yellow,circle,text=black,thick] (q_2) [above left of=q_0]  {$3$};
    \node[draw,blue,circle,fill=red!40,text=black,thick] (q_3) [below right of=q_0] {$4$};
    \node[draw,yellow,circle,text=black,thick] (q_4) [below left of=q_0]  {$5$};

    \path
    (q_0) edge (q_1) edge (q_2) edge (q_3) edge (q_4)
    (q_1) edge (q_2)
    (q_3) edge (q_4);
  \end{tikzpicture}}}\hspace{48pt}$\rightarrow$\hspace{48pt}
  \subfloat{
    \begin{tabular}{c|cccc}
      \multirow{6}{*}{$G$} & $\sigma_{1}$ & $\sigma_{2}$ & $\sigma_{6}$ & $\sigma_{6}$ \\
      &$\sigma_{1}$&$\sigma_{6}$&$\sigma_{3}$&$\sigma_{6}$\\
      &$\sigma_{1}$&$\sigma_{4}$&$\sigma_{6}$&$\sigma_{6}$\\
      &$\sigma_{1}$&$\sigma_{6}$&$\sigma_{5}$&$\sigma_{6}$\\
      &$\sigma_{6}$&$\sigma_{2}$&$\sigma_{3}$&$\sigma_{6}$\\
      &$\sigma_{6}$&$\sigma_{4}$&$\sigma_{5}$&$\sigma_{6}$\\
      \midrule
      $B$ & $\sigma_{6}$ & $\sigma_{6}$ & $\sigma_{6}$ & $\sigma_{6}$
    \end{tabular}
  }\\
  \subfloat{
    $P = \{\sigma_{1}\ast\ast\ast,\;\;\; \ast\sigma_{2}\ast\ast,\;\;\; \ast\sigma_{4}\ast\ast\}$
  }
  \caption{An example of the reduction used in Lemma~\ref{lem:VC<=PIS_short_strings} with $k=3$. The vertex cover is highlighted in red, and the correspond set of patterns is shown. Note the difference with the reduction in Lemma~\ref{lem:VC<=PI}, here the position encodes the coloring and the symbols encode the edges, whereas previously the string more directly encode the graph.}
  \label{fig:planar_VC_red_example}
\end{figure}

\begin{proof}
Let $(\mathfrak{G},k)$ be an instance of \PVC{}. We assume without loss of generality that $V(\mathfrak{G}) = \{1,\ldots, n\}$. As $\mathfrak{G}$ is planar, we can compute a proper 4-coloring in polynomial time~\cite{AppelHaken1989}. Let $C : V(\mathfrak{G}) \rightarrow \{1,2,3,4\}$ be such a coloring.
We construct an instance $(\Sigma, G, B, k', d)$ of \PIS{} as follows:
\begin{enumerate}
\item $\Sigma = \{\sigma_{1},\ldots,\sigma_{n+1}\}$.
\item $G = \{g_{ij} \mid ij \in E(\mathfrak{G})\}$ where for $k \in \{1,\ldots, 4\}$ we set 
\[
g_{ij}[k] :=
\begin{cases}
\sigma_{i} & \text{if } C(i) = k\\
\sigma_{j} & \text{if } C(j) = k\\
\sigma_{n+1} & \text{otherwise.}
\end{cases}
\]
\item $B = \{\sigma_{n+1}^{4}\}$.
\item $d = 4$.
\end{enumerate}

We note that as $C$ is a proper coloring, $C(i) \neq C(j)$ for any $ij \in E(\mathfrak{G})$. Moreover for $i \in V(\mathfrak{G})$, $\sigma_{i}$ only appears as the $C(i)$th symbol in any string.

The construction can clearly be performed in polynomial time. The construction is illustrated in Figure~\ref{fig:planar_VC_red_example}.

\begin{claim}
  If $(\mathfrak{G},k)$ is a \Yes{} instance of \PVC{} then $(\Sigma, G, B, k, d)$ is a \Yes{} instance of \PIS{}.
\end{claim}

Let $V' \subseteq V(\mathfrak{G})$ with $\Card{V'} \leq k$ be a vertex cover witnessing that $(\mathfrak{G},k)$ is a \Yes{} instance of \PVC{}. We construct a set $P$ with $\Card{P} = \Card{V'} \leq k$ of patterns that forms a solution for $(\Sigma, G, B, k, d)$ in the following manner: for each $i \in V'$, we add the pattern $p_{i}$ to $P$ where $p_{i}[C(i)] = \sigma_{i}$ and all other symbols in $p_{i}$ are $\ast$. 
No pattern in $P$ is compatible with the singleton element of $B$, as each has a symbol $\sigma_{i}$ with $1 \leq i \leq n$. For every edge $ij \in E(\mathfrak{G})$, at least one of $i$ and $j$ is in $V'$. Without loss of generality assume that $i \in V'$. By construction the string $g_{ij}$ is compatible with the pattern $p_{i} \in P$, therefore every string in $G$ is compatible with some pattern in $P$.

\begin{claim}
  If $(\Sigma, G, B, k, d)$ is a \Yes{} instance of \PIS{} then $(\mathfrak{G},k)$ is a \Yes{} instance of \PVC{}.
\end{claim}

Let $P$ with $\Card{P} \leq k$ be a set of patterns such that $P \to (G,B)$. As before we may assume that $P$ is minimal in the sense that each pattern is compatible with some string in $G$. Each $p \in P$ must have at least one symbol drawn from the set $\{\sigma_{1}, \ldots, \sigma_{n}\}$, otherwise $p\to B$. No pattern $p \in P$ can have more than two symbols from $\{\sigma_{1}, \ldots, \sigma_{n}\}$, otherwise $p \not\to G$. As before, we partition $P$ into $P_{1}$, the subset of patterns with one symbol from $\{\sigma_{1}, \ldots, \sigma_{n}\}$, and $P_{2}$, the subset of patterns with two symbols from $\{\sigma_{1}, \ldots, \sigma_{n}\}$. We construct a vertex cover $V' \subseteq V(\mathfrak{G})$ for $\mathfrak{G}$ with $\Card{V'} \leq \Card{P} \leq k$ as follows:
\begin{itemize}
\item for each $p \in P_{1}$ add $i$ to $V'$ if $p[C(i)] = \sigma_{i}$,
\item for each $p \in P_{2}$ where $p[C(j)] = \sigma_{j}$ and $p[C(i)] = \sigma_{i}$, arbitrarily add either $i$ or $j$ to $V'$.
\end{itemize}
Consider every edge $ij \in E(\mathfrak{G})$. The string $g_{ij}$ is compatible with some pattern $p \in P$, therefore at least one of $i$ and $j$ is in $V'$, thus $V'$ forms a proper vertex cover for $\mathfrak{G}$.
\end{proof}

\begin{proof}[Proof of Theorem~\ref{thm:PIS_NP-c_short_strings}]
The construction used in the proof of Lemma~\ref{lem:VC<=PIS_short_strings} has the required structural properties. Again containment in \NP{} is apparent from the usual verification algorithm techniques.
\end{proof}

\begin{corollary}\label{cor:PIS_paraNP-c_short_strings}
\sloppypar \PIS{} is \paraNP{}-complete when parameterized by $n + d + \Card{B}$.  
\end{corollary}

\begin{proof}
The corollary follows from Theorem~\ref{thm:PI_NP-c_small_strings} and Corollary 2.16 from~\cite{FlumGrohe2006}.
\end{proof}

\subsection{Containment}
\label{sec:containment}

Although the \Wtwo{}-hardness reduction is quite direct, containment of \PI{} when parameterized by $k$ is not apparent. In fact it is not clear that the problem lies in \WP{} or even \XP{}. As the non-parameterized version of the problem is \NP{}-complete, it is at least contained in \paraNP{}. For \PIp{} we have shown containment in \Wtwo{}. In contrast, for \PIP{} we can show containment in \Wstarfive{}.

\begin{theorem}\label{thm:PIP_in_W*[5]}
$\PIP{} \in \Wstarfive$ when parameterized by $k+r$.
\end{theorem}

\begin{proof}
We reduce the problem to \MCSigmaStarFiveOne{}, which is complete for \Wstarfive{}~\cite{ChenFlumGrohe2007,FlumGrohe2006}. We use the same first-order structure as in the proof of Corollary~\ref{cor:PIp_W[2]-complete}, and give a suitable first-order formula:

  \begin{align*}
    &\exists s_{1},\ldots,s_{k},i_{1,1},\ldots,i_{k,r}\forall j\\
    &(Gj \rightarrow (\exists l (\bigvee_{c \in [k]} l = s_{c} \wedge \forall b (Cjb = Clb \vee \bigvee_{d \in [r]} b = i_{c,d})))) \wedge\\
    &(Bj \rightarrow (\forall l (\bigwedge_{c \in [k]} l = s_{c} \wedge \exists b (Cjb \neq Clb \wedge \bigwedge_{d \in [r]} b \neq i_{c,d})))) \wedge\\
    &(\bigwedge_{c \in [k]} (Ns_{c} \wedge \bigwedge_{d \in [r]} Ni_{c,d}))\\
  \end{align*}

The formula picks out $k$ indices of strings (implicitly in $G$, as a choice of a string from $B$ will fail) and for each of these, $r$ indices which will be the location of the $\ast$ symbols in the patterns. For each index, if the index selects a string in $G$, then one of the patterns is compatible with the string, if it selects a string in $B$, no pattern is compatible with the string. We note that the $B$ clause is in $\Pi_{2,1}$, and hence $\Sigma_{3,1}$, giving the final bound of $\Sigma^{\ast}_{5,1}$.
\end{proof}

This also places \PIP{} somewhere between \Wfive{} and \Weight{}~\cite{ChenFlumGrohe2007}. We note that the above formula could be converted into prenex form, giving a tighter containment, however the central observation is that it will be greater than \Wtwo{}, in contrast to the hardness result and the containment of \PIp{}.

\section{Tractable Cases of Pattern Identification Problem}
\label{sec:easy}

Guided by the results of Section~\ref{sec:hard}, we identify the following cases where the \PI{} problem is tractable.

\begin{theorem}\label{thm:PI_FPT_k_Sigma_n}
\sloppypar \PI{} is fixed-parameter tractable when parameterized by $\Card{\Sigma} + n$.
\end{theorem}

\begin{proof}
Taking the alphabet size and the string length as a combined parameter gives an immediate kernelization. The total number of strings of length $n$ over alphabet $\Sigma$ is $\Card{\Sigma}^{n}$. Thus $\Card{G} + \Card{B} \leq \Card{\Sigma}^{n}$.
\end{proof}

\begin{theorem}\label{PI_FPT_few_small_strings}
\sloppypar \PIS{} is fixed-parameter tractable when parameterized by $d + \Card{G} + \Card{B}$, with a kernel of size $O(d\cdot(\Card{G} + \Card{B})^{2})$ in both the total number of symbols across all strings and the size of the alphabet.  
\end{theorem}

\begin{proof}
 As $G$ and $B$ are $d$-small, there can be at most $d\cdot(\Card{G} + \Card{B})$ positions where any pair of strings in $G \cup B$ differ, that is, every other position must be the base symbol uniformly across all strings. The positions where every string is identical cannot be of use in generating patterns, thus we may ignore these positions. This gives restricted sets $G'$ and $B'$ of size $\Card{G'} + \Card{B'} \leq \Card{G} + \Card{B}$ containing strings of length at most $d\cdot(\Card{G} + \Card{B})$. Furthermore this restricts the number of symbols used from $\Sigma$ to at most $d\cdot(\Card{G} + \Card{B})^{2}$. Thus we can restrict our alphabet to these symbols alone, denote this new alphabet by $\Sigma'$. This gives our required kernel size.

The initial determination of which positions to ignore can be computed in $O(n \cdot(\Card{G} + \Card{B}))$ time, thus the kernelization can be performed in polynomial time.
\end{proof}

\begin{theorem}\label{thm:PI_FPT_k_n}
\sloppypar \PI{} is fixed-parameter tractable when parameterized by $k + n$.
\end{theorem}

\begin{proof}
Let $(\Sigma, G, B, k)$ be an instance of \PI{}. If $(\Sigma, G, B, k)$ is a \Yes{} instance, by definition, there exists a $P$ with $\Card{P} \leq k$ such that every string $g \in G$ must be compatible with at least one $p \in P$. Therefore given $g$, the compatible $p$ must consist of, at each position, either the $\ast$ symbol, or the symbol at the same position in $g$.

\sloppypar This gives a direct bounded search tree algorithm for $\PI{}$. At each node in the tree we select an arbitrary $g$ from $G$. We then branch on all possible patterns $p$ that are compatible with $g$, with a new set $G :=  G \setminus \{h \in G \mid p \to h\}$ (note that this removes $g$ from further consideration). If there is a $b \in B$ such that $p \to b$, then we terminate the branch. If we reach depth $k$ and $G \neq \emptyset$, we terminate the branch. Otherwise if at any point we have $G = \emptyset$, we answer \Yes{}.

Obviously the depth of the search tree is explicitly bounded by $k$. The branching factor is equal to the number of patterns compatible with a string of length $n$, which is $2^{n}$. The adjustment of $G$ and checks against $B$ at each node individually take $O(n)$ time, giving $O((\Card{G} + \Card{B})\cdot n)$ time at each node. Combined the algorithm takes $O(2^{kn}\cdot(\Card{G} + \Card{B})\cdot n)$ time, and the theorem follows.

\end{proof}

\begin{theorem}\label{thm:PI_FPT_G_n}
  \sloppypar \PI{} is fixed-parameter tractable when parameterized by $\Card{G} + n$.
\end{theorem}

\begin{proof}
The search tree approach used in the proof of Theorem~\ref{thm:PI_FPT_k_n} can also be adapted to the combined parameter $\Card{G} + n$. Again we select an arbitrary $g$ from $G$. We branch on all possible patterns $p$ that are compatible with $g$, of which there are at most $2^{n}$, with the new set $G :=  G \setminus \{h \in G \mid p \to h\}$. If $p \to b$ for any $b \in B$, the branch is terminated. When we have $G = \emptyset$, we check whether the collected set $P$ of patterns in that branch. If $\Card{P} \leq k$ we answer \Yes{}, otherwise the branch is terminated. If all branches terminate with no \Yes{} answer, we answer \No{}.
\end{proof}

\begin{theorem}\label{thm:PIP_FPT_k_r_Sigma_B}
\sloppypar \PIPS{} is fixed-parameter tractable when parameterized by $k + \Card{\Sigma} + d + r + \Card{B}$.
\end{theorem}

\begin{proof}
As each pattern can have at most $r$ many $\ast$ symbols, every other symbol in each pattern is fixed. Thus each pattern is compatible with $\Card{\Sigma}^{r}$ strings. This limits the number of strings in $G$ to $k\cdot\Card{\Sigma}^{r}$.

The tractability then follows from Theorem~\ref{PI_FPT_few_small_strings}.
\end{proof}

\section{Discussion}
\label{sec:discussion}

Complementing the classification results given above, we now discuss some related issues. Firstly (in Section~\ref{sec:discussion_parameters}), given the complex parameter landscape introduced, what problems remain unsolved, and which are the interesting parameterizations for future work? Secondly, we related \PI{} to some similar problems that give some small intuition as to sources of complexity in \PI{} (Section~\ref{sec:other_problems}).

\subsection{The Mire of Multivariate Analysis: Untangling the Parameters}
\label{sec:discussion_parameters}

The complexity analysis in this work involves a considerable number of parameters and unsurprisingly, there are some relationships between them that can be identified, allowing a better perspective on the sources of complexity in the problem, and what cases remain open. The immediately obvious relationships, for non-trivial parameter values\footnote{By non-trivial we mean values which differentiate the parameters -- for example, if $s > n$, $s$ becomes meaningless as any number of $\ast$ symbols would be allowed, within the limitation of length $n$ strings.}, are $r \leq n$, $s \leq n$ and $d \leq n$. We also note that $k \leq \Card{G}$ and $k \leq (\Card{\Sigma}+1)^{n}$, again for non-trivial values of $k$.

This helps to unravel some of the relationships present in the results of this work. We also note that, of course, expanding a list of parameters preserves tractability, while reducing a list of parameters preserves intractability

\begin{figure}[htb]
  \centering
  \includegraphics[scale=0.8]{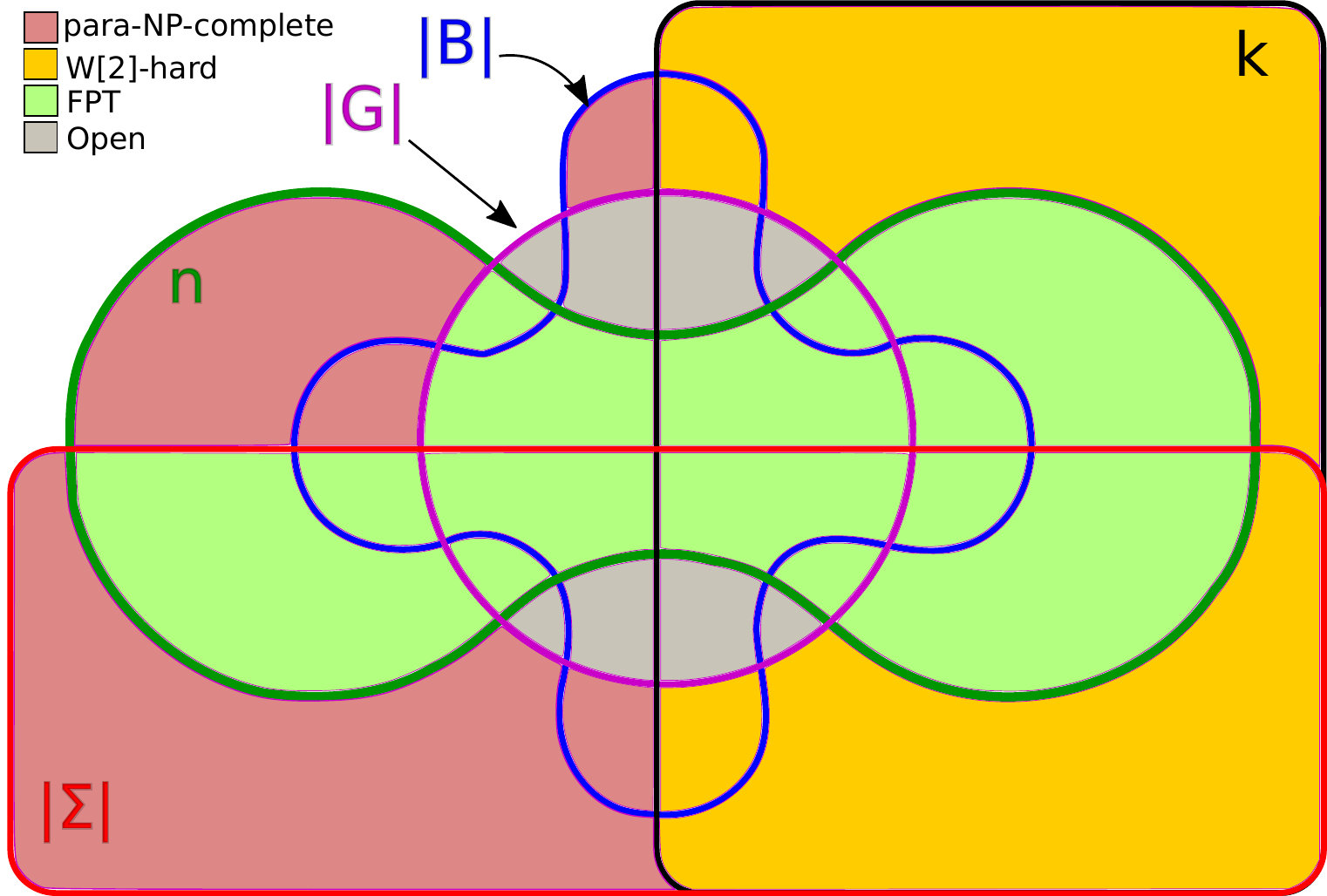}
  \caption{Simplified representation of the parameter space and the complexity results. We note in particular that $n$ or at least one of its related parameters $s$, $r$ or $d$ seems essential for tractability (though never sufficient). Given the nature of the input as a set of strings, it is perhaps unsurprising that at least two parameters are (apparently) needed for tractability. The obvious open cases are dominated by the parameter $\Card{G}$.}
  \label{fig:patterns-venn}
\end{figure}

A visual summary of the tractable, intractable and open cases for a simplified parameter space is given in Figure~\ref{fig:patterns-venn}. Given the relationships between $s$, $r$, $d$ and $n$, we reduce the parameter space to $k$, $\Card{\Sigma}$, $n$, $\Card{G}$ and $\Card{B}$. Although this reduces the accuracy of the space, the broad complexity landscape of the problem becomes more comprehensible.

Speculatively, we may observe that the problem seems to require at least two parameters for tractability. This is perhaps unsurprising, given the nature of the input -- we need some parametric ``handle'' on the length of the strings and another on the number of strings.

From Figure~\ref{fig:patterns-venn} it is clear that the central parameter in the open cases is $\Card{G}$, though we note that in the full parameter space, there are combinations of $s$, $r$ and $d$ with other parameters for which the complexity remains open\footnote{At last hand-count, 72 cases out of the 256 possible parameterizations with these parameters remain open, compared to 8 with the reduced parameter space.}.

\subsection{Ties to Other Problems}
\label{sec:other_problems}

The \PI{} problem, as would be expected, has ties to other problems that (can) model the general search for patterns that separate two sets of data. These ties also illustrate some features of the computational complexity of the problem.

\subsubsection{Set Covering}
\label{sec:set_covering}

When the length $n$ of the strings is small, \PI{} can be easily reduced to \SetCover{}. Given an instance $(\Sigma, G, B, k)$ of \PI{}, we can generate the set $P$ of all patterns that are compatible with some string in $G$. We know that $\Card{P} \leq \Card{G}\cdot 2^{n}$. From $P$ we remove any pattern that is compatible with a string in $B$. Let $P'$ be the set thus obtained. For each $p \in P'$, let $s_{p} = \{g \in G \mid p \to g\}$, and let $S = \{s_{p} \mid p \in P'\}$. Taking $G$ as the base set, $(G,S,k)$ forms an instance of \SetCover{} (parameterized by $k$). This reduction can be performed in $O((\Card{B}+\Card{G})\cdot\Card{G}\cdot 2^{n}n)$ time.

This leads to the following theorem:

\begin{theorem}\label{thm:PI_in_W[2]_sometimes}
  $\PI{} \in \Wtwo{}$ when $n \leq f(k)\cdot\log\Card{I}$ where $\Card{I}$ is the overall size of the instance and $f(k)$ is a computable function of the parameter $k$.
\end{theorem}

\begin{proof}
  The reduction above is a parameterized reduction if $2^{n} \in O(g(k)\cdot\Card{I}^{c})$ for some computable function $g$.
\end{proof}

It is not clear that we retain \Wtwo{}-hardness in this case however, so we unfortunately do not obtain a \Wtwo{}-completeness result.

This does give us an immediate approximation algorithm for this case however. As \SetCover{} has a $1 + \log(\Card{S})$-factor linear time approximation algorithm~\cite{Johnson1974}, we obtain a $1 + \log(\Card{G}^{2}\cdot\log(\Card{I})\cdot 2^{f(k)})$-factor fpt-time approximation algorithm. 

\subsubsection{Feature Set}
\label{sec:feature_set}

The \kFS{} problem bears a strong resemblance to the \PI{} problem\footnote{Indeed, variants of \kFS{} have also been considered for use in similar applications as \PI{}~\cite{CottaMoscato2005}.}, except in the \kFS{} case, the problem asks for a set of features that separate the ``good'' examples from the ``bad'' rather than a set of patterns. In fact, given a feasible solution for one problem, we can construct a feasible (but not necessarily optimal) solution to the other.

Given a set $I = \{i_{1}, \ldots, i_{k}\}$ of indices of columns forming a feature set, we can construct a set of patterns that separates $G$ and $B$ as follows: for each $g \in G$, let $p_{g}$ be the pattern where $p_{g}[i] = g[i]$ if $i \in I$ and $p_{g}[i] = \ast$ otherwise. We note that this gives a set of small patterns (\emph{i.e.}, $s = k$), however the number of patterns may be as large as $\Card{G}$.

Conversely, given a set of patterns $P$ with at most $s$ non-$\ast$ symbols in each pattern, the set $I = \{i \in [n] \mid \exists p \in P(p[i] \neq \ast)\}$ forms a feature set. Again we note that the size of the feature set may be as large as $\Card{G}\cdot s$.

If we consider a variant of $\PIp$ where we relax the constraint on the number of patterns in the solution, it is easy to see that this problem is in \Wtwo{}. This suggests that the solution size plays a significant role in raising the complexity of the problem from a parameterized perspective.


\section{Conclusion and Future Directions}
\label{sec:conclusion}

There are a number of open complexity questions prompted by this paper, three of which we think are particularly interesting.

The central question is of course the precise classification of \PI{}. Although \PIp{} is \Wtwo{}-complete, the general problem is only \Wtwo{}-hard, and the containment of \PIP{} simply gives a loose upper bound, although does suggest that the problem is harder than \PIp{}. The problem, intuitively, also shares some similarities with $p$-\textsc{Hypergraph-(Non)-Dominating-Set} which is \Wthree{}-complete~\cite{ChenFlumGrohe2007}. $p$-\textsc{Colored-Hypergraph-(Non)-Dominating-Set} however is $\mathsf{W}^{\ast}[\mathsf{3}]$-complete~\cite{ChenFlumGrohe2007} and appears ``harder'' than \PI{}, hence we conjecture:

\begin{conjecture}\label{conj:PI_in_W[2]}
\sloppypar $\PI{}$ is $\Wthree{}$-complete when parameterized by $k$.
\end{conjecture}

There are also some interesting parameterizations for which the complexity remains open:

\begin{itemize}
\item \sloppypar \PIS{} parameterized by $k + \Card{\Sigma} + d$, and
\item \sloppypar \PIPS{} parameterized by $k + d + r$.
\end{itemize}

Turning to the parameter $\Card{G}$, results for the following combinations of parameters would also close some of the significant open cases:

\begin{itemize}
\item \sloppypar \PI{} parameterized by $k + \Card{\Sigma} + \Card{G}$,
\item \PI{} parameterized by $\Card{G} + \Card{\Sigma} + \Card{B}$, and
\item \PI{} parameterized by $k + \Card{B} + \Card{G}$.
\end{itemize}

As a matter of prognostication, we would guess that the first of these is in \FPT{}, and the latter two are hard for some level of the \W{}-hierarchy, but as yet have no strong evidence for these claims.

\section{Acknowledgements}
\label{sec:acknowledgements}

PM acknowledges funding of his research by the Australian Research Council (ARC, http://www.arc.gov.au/) grants Future Fellowship FT120100060 and Discovery Project DP140104183.

\section{References}
\label{sec:references}

\bibliographystyle{plain}
\bibliography{patterns}

\end{document}